\documentclass[11pt,draft]{llncs}

\usepackage{amssymb,amsmath,enumerate,graphicx,hyperref,amsfonts,latexsym,url}
\usepackage[utf8]{inputenc}
\usepackage[T1]{fontenc}        
\usepackage[normalem]{ulem}

\usepackage{tikz}
\usetikzlibrary{snakes,shapes}
\tikzstyle{hyb}=[rectangle,draw,minimum size=3mm,inner sep=0.2pt]
\tikzstyle{tre}=[circle,draw,minimum size=3mm,inner sep=0.2pt]
\tikzstyle{treg}=[circle,draw,minimum size=4mm,inner sep=0.2pt]
\newcommand{\etq}[1]{%
\draw (#1) node {\tiny $#1$};
}

\newcommand{\RR}{\mathbb R}

\renewcommand{\leq}{\leqslant}
\renewcommand{\geq}{\geqslant}

\begin{document}

\title{The Fair Proportion is  a  Shapley Value  on phylogenetic networks too}

\author{Tom\'as M. Coronado \and Gabriel Riera \and Francesc Rossell\'o}

\institute{Dept. of Mathematics and Computer Science, University of the Balearic Islands, E-07122 Palma, Spain, and Balearic Islands Health Research Institute (IdISBa), E-07010 Palma, Spain
\texttt{$\mathtt{\{}$t.martinez,gabriel.riera,cesc.rossello$\mathtt{\}}$@uib.edu}}
\maketitle

\begin{abstract}
The Fair Proportion of a species in a phylogenetic tree is a very simple measure that has been used to assess its value relative to the overall phylogenetic  diversity represented by the tree. It has recently been proved by Fuchs and Jin to be equal to the Shapley Value of the coallitional game that sends each subset of species to its rooted Phylogenetic Diversity in the tree. We prove in this paper that this result extends to the natural translations of the Fair Proportion and the rooted Phylogenetic Diversity to rooted phylogenetic networks. We also generalize to rooted phylogenetic networks the expression for the Shapley Value of the unrooted Phylogenetic Diversity game on a phylogenetic tree established by Haake, Kashiwada and Su. \end{abstract}

\section{Introduction}

An important problem in ecology is to assess the genetic value of individual species, with the aim of ranking them for conservation prioritization purposes \cite{Diniz}. One of the simplest measures proposed in this connection is the Fair Proportion of a species in a phylogenetic tree, introduced by  Redding and Mooers in \cite{RM}. This index apportions  the overall diversity of a phylogenetic tree among its leaves  by equally dividing the weight of each arc  among its descendant leaves. Although this index is  very easy to define, it is not obvious at first sight that it defines a sound and meaningful ranking of the species' genetic value. On the other hand, the Shapley Value of a species in a  phylogenetic tree, introduced by  Haake, Kashiwada and Su in \cite{Haake1}, which is based on a well-known solution from cooperative game theory to the problem of dividing the global value of a game among its players,
 lies at the other end of the individual biodiversity measures spectrum, in the sense that it provides a meaningful distribution of the global diversity of a phylogenetic tree among its leaves at the cost of being defined through quite a complex formula,  involving a sum of an exponential number of terms. 
But, in what Steel dubs as an ``interesting and not immediately obvious'' result \cite[p. 141]{Steelbook},  Fuchs and Jin proved in  \cite{FJ} that Fair Proportions and Shapley Values are exactly the same on phylogenetic trees, thus yielding an individual biodiversity index  which is easy to define and compute and which ranks species in a very clear and meaningful way.

In this note we extend Fuchs and Jin's result from rooted phylogenetic trees to rooted phylogenetic networks \cite{Huson},  graphical models of evolutionary histories that allow the inclusion of reticulate processes like hybridizations, recombinations or lateral gene transfers. More specifically, we show that if we define the Fair Proportion of a leaf in a rooted phylogenetic network exactly as if we were a phylogenetic tree ---we split the weight of each arc equally among all its descendant leaves,  and then we add up the leaf's share of the weights of all its ancestor arcs--- then it is equal to the subnet Shapley Value of the leaf in the network as defined by Wicke and Fischer in \cite{WF2}. We also extend to rooted phylogenetic networks 
the simple expression for the unrooted phylogenetic Shapley Value established by  Haake, Kashiwada and Su in \cite{Haake1}, thus showing in particular that it can  be computed efficiently also on rooted phylogenetic networks.

\section{Preliminaries}
Let  $\Sigma$ be a finite set of labels. A \emph{$\Sigma$-rDAG} is a rooted directed acyclic graph  with its \emph{leaves} (its nodes of out-degree 0) bijectively labeled in $\Sigma$. 
We shall denote the sets of nodes and arcs of a $\Sigma$-rDAG $N$ by $V(N)$ and $E(N)$, respectively, and  we shall always identify its leaves  with their  corresponding labels.  A \emph{weighted} $\Sigma$-rDAG 
is a $\Sigma$-rDAG endowed with a mapping $\omega:E(N) \to [0,\infty)$ that assigns a weight $\omega(e)\geq 0$ to every arc~$e$.

Given two nodes $u,v$ in a $\Sigma$-rDAG $N$, we say that $v$ is a \emph{child} of $u$, and also that $u$ is a \emph{parent} of $v$, when $(u,v)\in E(N)$, and that $v$ is a  \emph{descendant} of $u$, and also that $u$ is an  \emph{ancestor} of $v$, when there exists a directed path from $u$ to $v$ in $N$. The \emph{cluster} $C(e)$ of $e\in E(N)$ is the set of descendant leaves of its end, and  we shall denote by $\kappa(e)$  the cardinal of $C(e)$. If $a\in C(e)$, we shall also say that $e$ is an \emph{ancestor arc} of $a$.

A \emph{phylogenetic network} on $\Sigma$ is a $\Sigma$-rDAG without \emph{elementary} nodes (that is, without nodes of in-degree $\leq 1$ and out-degree 1).   A node in a phylogenetic network is of \emph{tree type} when its in-degree is 0 (the \emph{root}) or 1, and a 
 \emph{reticulation} when its in-degree is at least 2. An arc is  of \emph{tree type} (respectively, of \emph{reticulate type}) when its end  is a tree node (resp., in a reticulation). 
Given a phylogenetic network $N$ on $\Sigma$ and a subset $X\subseteq \Sigma$, we shall denote by  $N(X)$  the subgraph of $N$ induced by the set of  all the ancestors of the leaves in $X$: it is a $X$-rDAG,  with the same root as $N$.

A \emph{phylogenetic tree} is a phylogenetic network without reticulations. Let us emphasize, hence, that all our phylogenetic trees are  rooted, unless otherwise explicitly stated.
 Given a weighted phylogenetic tree  $T$ on $\Sigma$, for every $a\in \Sigma$ and for every $X\subseteq \Sigma$:
\begin{itemize}
\item The \emph{Fair Proportion} of $a$ in $T$   {\cite{RM}} is 
$$
FP_T(a)=\sum_{e:\, a\in C(e)}  \frac{\omega(e)}{\kappa(e)}.
$$ 

\item The  \emph{rooted Phylogenetic Diversity} $rPD_T(X)$ of $X$ in $T$ \cite{Faith92} is the \emph{total weight} of $T(X)$, that is, the sum of the weights of its arcs:
$$
rPD_T(X)=\sum_{e:\, X\cap C(e)\neq \emptyset} \omega(e).
$$

\item The \emph{unrooted Phylogenetic Diversity} $uPD_T(X)$ of $X$ in $T$ \cite{Faith92} is the total weight of the smallest unrooted subtree of $T$ containing the leaves in $X$, or equivalently, the total weight of the subtree of $T(X)$ rooted at the lowest common ancestor $LCA_T(X)$ of $X$. 
\end{itemize}

A \emph{coalitional game} on a set $\Sigma$ is simply a set function
$W:\mathcal{P}(\Sigma)\to \RR$. 
For every $a\in \Sigma$, the \emph{Shapley Value} on $a$ of a coalitional game $W$ on $\Sigma$ \cite{Shapley} is a weighted average of the marginal contribution of $a$ to the value, under $W$, of each coalition $X\subseteq \Sigma$:
$$
SV_{a}(W) =\sum_{a\in X\subseteq \Sigma}\frac{(|X|-1)!(|\Sigma|-|X|)!}{|\Sigma|!}\big(W(X)-W(X\setminus\{a\})\big)
$$
The \emph{Shapley value} of the game $W$ is then the vector $(SV_{a}(W))_{a\in \Sigma}$.

\section{The Fair Proportion is a Shapley Value}

Let  $N$ be a weighted phylogenetic network on $\Sigma$. 
We define the \emph{Fair Proportion}  of  $a\in \Sigma$ in $N$, $FP_N(a)$, exactly as if $N$ were a phylogenetic tree: we split the weight of each arc equally among all its descendant leaves,  and then we add up $a$'s share of the weights of all its ancestor arcs:
$$
FP_N(a)=\sum_{e:\, a\in C(e)}  \frac{\omega(e)}{\kappa(e)}.
$$ 
In particular, if $N$ is a phylogenetic tree, this Fair Proportion is equal to the one defined on phylogenetic trees by Redding and Mooers and recalled in the previous section.
Our goal in this section is to show that, as it already happens on phylogenetic trees \cite{FJ},
 this Fair Proportion  is the Shapley Value of a certain ``phylogenetic diversity'' coallitional game: namely, of   $rPSD_N: \mathcal{P}(\Sigma)\to \RR$ on $\Sigma$, where, for every  $X\subseteq \Sigma$, $rPSD_N(X)$ is the 
 \emph{rooted Phylogenetic Subnet Diversity} of $X$ in $N$ in the sense of  \cite[Def. 6]{WF2}, which is defined as the total weight of $N(X)$:
$$
rPSD_N(X)=\sum_{e:\, X\cap C(e)\neq \emptyset} \omega(e);
$$
in particular,  $rPSD_N(\emptyset)=0$.
Notice that if $T$ is a phylogenetic tree, then $rPSD_T$ is equal to Faith's rooted Phylogenetic Diversity $rPD_T$ recalled in the previous section.

For every $a\in \Sigma$, let its \emph{rooted subnet Shapley Value} in $N$ be the Shapley Value of $rPSD_N$ on $a$:
$$
SV_N(a)=\sum_{a\in X\subseteq \Sigma} \frac{(|X|-1)!(|\Sigma|-|X|)!}{|\Sigma|!}\big(rPSD_N(X)-rPSD_N(X\setminus \{a\})\big).
$$

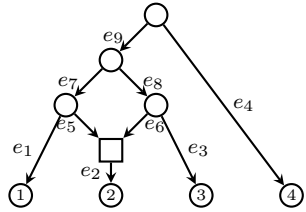
\begin{figure}[htb]
\begin{center}
\begin{tikzpicture}[thick,>=stealth,scale=0.3]
\draw(0,0) node[tre] (1) {}; \etq{1}
\draw(4,0) node[tre] (2) {}; \etq{2}
\draw(8,0) node[tre] (3) {}; \etq{3}
\draw(12,0) node[tre] (4) {}; \etq 4
\draw(4,2) node[hyb] (h1) {}; 
\draw(2,4) node[tre] (a) {}; 
\draw(6,4) node[tre] (b) {}; 
\draw(4,6) node[tre] (c) {}; 
\draw(6,8) node[tre] (r) {}; 
\draw[->](r) --node[midway, left]{\scriptsize $e_9$}(c);
\draw[->](r)--node[midway, right]{\scriptsize $e_4$}(4);
\draw[->](c)--node[midway, left]{\scriptsize $e_7$}(a);
\draw[->](c)--node[midway, right]{\scriptsize $e_8$}(b);
\draw[->](a)--node[midway, left]{\scriptsize $e_1$}(1);
\draw[->](a)--node[midway, left]{\scriptsize $e_5$}(h1);
\draw[->](h1)--node[midway, left]{\scriptsize $e_2$}(2);
\draw[->](b)--node[midway, right]{\scriptsize $e_6$}(h1);
\draw[->](b)--node[midway, right]{\scriptsize $e_3$}(3);
\end{tikzpicture}
\end{center}
\caption{\label{fig:1} The phylogenetic network used in Example \ref{ex:1}.}
\end{figure}

\begin{example}\label{ex:1}
Consider the phylogenetic network $N$ depicted in Figure \ref{fig:1} and let $w_i=\omega(e_i)$, for every $i=1,\ldots,9$. Then:
$$
\begin{array}{l}
\kappa(e_1)=\kappa(e_2)=\kappa(e_3)=\kappa(e_4)=\kappa(e_5)=\kappa(e_6)=1\hspace*{1.5cm}\\
\kappa(e_7)=\kappa(e_8)=2,\ \kappa(e_9)=3\\[1ex]
rPSD_N(1)=w_1+w_7+w_9\\
rPSD_N(2)=w_2+w_5+w_6+w_7+w_8+w_9\\
rPSD_N(3)=w_3+w_8+w_9\\
rPSD_N(4)=w_4\\
rPSD_N(1,4)=w_1+w_4+w_7+w_9\\
\end{array}
$$

$$
\begin{array}{l}
rPSD_N(2,3)=w_2+w_3+w_5+w_6+w_7+w_8+w_9\\
rPSD_N(2,4)=w_2+w_4+w_5+w_6+w_7+w_8+w_9\\
rPSD_N(3,4)=w_3+w_4+w_8+w_9\\
rPSD_N(1,2,3)=w_1+w_2+w_3+w_5+w_6+w_7+w_8+w_9\\
rPSD_N(1,2,4)=w_1+w_2+w_4+w_5+w_6+w_7+w_8+w_9\\
rPSD_N(1,3,4)=w_1+w_3+w_4+w_7+w_8+w_9\\
rPSD_N(2,3,4)=w_2+w_3+w_4+w_5+w_6+w_7+w_8+w_9\\
\end{array}
$$
So, the Shapley Values of the leaves of $N$ are:
$$
\begin{array}{l}
SV_N(1)\displaystyle =\frac{1}{4}\big(rPSD_N(1)-rPSD_N(\emptyset)\big)\\[1.5ex]
\displaystyle\qquad+\frac{1}{12}\big(rPSD_N(1,2)-rPSD_N(2)+rPSD_N(1,3)-rPSD_N(3)\\[1.5ex]
\displaystyle\qquad\qquad\quad+rPSD_N(1,4)-rPSD_N(4)\big)\\[1.5ex]
\displaystyle\qquad +\frac{1}{12}\big(rPSD_N(1,2,3)-rPSD_N(2,3)+rPSD_N(1,2,4)\\[1.5ex]
\displaystyle\qquad\qquad\quad-rPSD_N(2,4)+rPSD_N(1,3,4)-rPSD_N(3,4)\big)\\[1.5ex]
\displaystyle\qquad +
\frac{1}{4}\big(rPSD_N(1,2,3,4)-rPSD_N(1,2,3)\big)\\[1.5ex]
\displaystyle\quad =\frac{1}{4}(w_1+w_7+w_9)+\frac{1}{12}(3w_1+2w_7+w_9)+\frac{1}{12}(3w_1+w_7)+\frac{1}{4}w_1\\[1.5ex]
\displaystyle\quad =w_1+\frac{1}{2}w_7+\frac{1}{3}w_9=\frac{\omega(e_1)}{\kappa(e_1)}+\frac{\omega(e_7)}{\kappa(e_7)}+\frac{\omega(e_9)}{\kappa(e_9)}=FP_N(1)
\\[2ex]
\displaystyle SV_N(3)=w_3+\frac{1}{2}w_8+\frac{1}{3}w_9=\frac{\omega(e_3)}{\kappa(e_3)}+\frac{\omega(e_8)}{\kappa(e_8)}+\frac{\omega(e_9)}{\kappa(e_9)}=FP_N(3)\\[1.5ex]
\qquad\mbox{ (by symmetry)}
\\[2ex]
SV_N(2)\displaystyle =\frac{1}{4}\big(rPSD_N(2)-rPSD_N(\emptyset)\big)\\[1.5ex]
\displaystyle\qquad+\frac{1}{12}\big(rPSD_N(1,2)-rPSD_N(1)+rPSD_N(2,3)\\[1.5ex]
\displaystyle\qquad\qquad\quad-rPSD_N(3)+rPSD_N(2,4)-rPSD_N(4)\big)\\[1.5ex]
\displaystyle\qquad +\frac{1}{12}\big(rPSD_N(1,2,3)-rPSD_N(1,3)+rPSD_N(1,2,4)\\[1.5ex]
\displaystyle\qquad\qquad\quad-rPSD_N(1,4)+rPSD_N(2,3,4)-rPSD_N(3,4)\big)\\[1.5ex]
\displaystyle\qquad+
\frac{1}{4}\big(rPSD_N(1,2,3,4)-rPSD_N(1,3,4)\big)\\[1.5ex]
\end{array}
$$
$$
\begin{array}{l}
\displaystyle\quad =\frac{1}{4}(w_2+w_5+w_6+w_7+w_8+w_9)\\[1.5ex]
\displaystyle\qquad +\frac{1}{12}(3w_2+3w_5+3w_6+2w_7+2w_8+w_9)\\[1.5ex]
\displaystyle\qquad
+\frac{1}{12}(3w_2+3w_5+3w_6+w_7+w_8)+\frac{1}{4}(w_2+w_5+w_6)\\[1.5ex]
\displaystyle \quad
=w_2+w_5+w_6+\frac{1}{2}w_7+\frac{1}{2}w_8+\frac{1}{3}w_9\\[1.5ex]
\displaystyle \quad= \frac{\omega(e_2)}{\kappa(e_2)}+\frac{\omega(e_5)}{\kappa(e_5)}+\frac{\omega(e_6)}{\kappa(e_6)}+\frac{\omega(e_7)}{\kappa(e_7)}+\frac{\omega(e_8)}{\kappa(e_8)}+\frac{\omega(e_9)}{\kappa(e_9)}=FP_N(2)
\\[2ex]
\end{array}
$$
$$
\begin{array}{l}
SV_N(4)\displaystyle =\frac{1}{4}\big(rPSD_N(4)-rPSD_N(\emptyset)\big)\\[1.5ex]
\displaystyle \qquad +\frac{1}{12}\big(rPSD_N(1,4)-rPSD_N(1)+rPSD_N(2,4)\\[1.5ex]
\displaystyle\qquad\qquad\quad-rPSD_N(2)+rPSD_N(3,4)-rPSD_N(3)\big)\\[1.5ex]
\displaystyle\qquad +\frac{1}{12}\big(rPSD_N(1,2,4)-rPSD_N(1,2)+rPSD_N(1,3,4)\\[1.5ex]
\displaystyle\qquad\qquad\quad -rPSD_N(1,3)+rPSD_N(2,3,4)-rPSD_N(2,3)\big)\\[1.5ex]
\displaystyle\qquad+
\frac{1}{4}\big(rPSD_N(1,2,3,4)-rPSD_N(1,2,3)\big)\\[1.5ex]
\displaystyle\quad =\frac{1}{4}w_4
+\frac{1}{12}\cdot 3w_4
+\frac{1}{12}\cdot 3w_4
+\frac{1}{4}w_4
=w_4=\frac{\omega(e_4)}{\kappa(e_4)}=FP_N(4)
\end{array}
$$
\end{example}

In the simple phylogenetic network considered in the previous example, the subnet Shapley Value of each leaf was equal to its Fair Proportion. Next theorem establishes that it is always the case.

\begin{theorem}\label{th:main}
For every weighted phylogenetic network $N$ on $\Sigma$ and for every $a\in \Sigma$,
$$
FP_N(a) = SV_N(a).
$$
\end{theorem}

\begin{proof}
Set $|\Sigma|=n$. For every $X\subseteq \Sigma$ containing $a$,
$$
\begin{array}{l}
\displaystyle rPSD_N(X)-rPSD_N(X\setminus \{a\})  = 
\sum_{e:\, X\cap C(e)\neq \emptyset}\hspace*{-1ex} \omega(e)-\hspace*{-2ex}\sum_{e:\, (X\setminus \{a\})\cap C(e)\neq \emptyset}\hspace*{-2ex} \omega(e)\\[3ex]  \qquad\qquad \displaystyle =\sum_{e:\, X\cap C(e)=\{a\}}\hspace*{-1ex}\omega(e)
\end{array}
$$
Then,
$$
SV_N(a)=\sum_{k=1}^n\frac{(k-1)!(n-k)!}{n!}\hspace*{-2ex}\sum_{|X|=k, a\in X} \big(rPSD_N(X)-rPSD_N(X\setminus \{a\})\big)
$$
where
$$
\begin{array}{l}
\displaystyle \sum_{|X|=k, a\in X} (rPSD_N(X)-rPSD_N(X\setminus \{a\}))=\hspace*{-2ex}
\sum_{|X|=k, a\in X}\sum_{e:\, X\cap C(e)=\{a\}}\omega(e)\\[2ex]
\displaystyle\qquad\qquad
=\sum_{e:\, a\in C(e)}|\{Y\subseteq \Sigma\setminus C(e)\mid |Y|=k-1\}|\cdot \omega(e)\\[2ex]
\displaystyle\qquad\qquad
=\sum_{e:\, a\in C(e)}\binom{n-\kappa(e)}{k-1}\omega(e)
\end{array}
$$
and therefore
$$
\begin{array}{l}
SV_N(a)  \displaystyle =\sum_{k=1}^n\Bigg(\frac{(k-1)!(n-k)!}{n!}\sum_{e:\, a\in C(e)}\binom{n-\kappa(e)}{k-1}\omega(e)\Bigg)\\[2ex]
\qquad  \displaystyle =\sum_{e:\, a\in C(e)}\Bigg(\omega(e)\sum_{k=1}^n\frac{(k-1)!(n-k)!}{n!}\binom{n-\kappa(e)}{k-1}\Bigg)\\[2ex]
 \qquad  \displaystyle =\sum_{e:\, a\in C(e)}\Bigg(\omega(e)\sum_{j=0}^{n-1}\frac{j!(n-j-1)!}{n!}\binom{n-\kappa(e)}{j}\Bigg)  =\sum_{e:\, a\in C(e)}\frac{\omega(e)}{\kappa(e)}
 \end{array}
$$
where the last equality is a consequence of Lemma 6.15 in \cite{Steelbook}, which establishes that, for every $1\leq m\leq n$,
$$
\sum_{j=0}^{n-1} \frac{j!(n-j-1)!}{n!}\binom{n-m}{j}=\frac{1}{m}.
$$
\qed
\end{proof}

\begin{remark}
A \emph{multilabelled tree} (a \emph{MUL-tree}, for short) on $\Sigma$ is a rooted tree with its leaves labelled in $\Sigma$. The difference with usual phylogenetic trees is that  the leaf labelling in a MUL-tree need not be bijective and,  thus, more than one leaf may be assigned the same label. MUL-trees include  \emph{area cladograms} \cite{ganapathy.ea:TCBB06} and \emph{gene trees} \cite{Gregg}.  Given a MUL-tree $T$, if, for every label $a\in \Sigma$ assigned to more than one leaf, we remove all leaves labelled with $a$ and the arcs ending in them and we add a new reticulation $h_a$, a new leaf labelled with $a$, new arcs from the parents of former leaves labelled with $a$ to $h_a$ and a new arc $(h_a,a)$, we obtain a phylogenetic network uniquely determined by $T$, which we dub  \emph{associated} to $T$. For instance, the phylogenetic network in Figure \ref{fig:1} is the associated to the MUL-tree depicted in Figure \ref{fig:multree}.
 
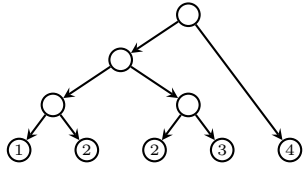
\begin{figure}[htb]
\begin{center}
\begin{tikzpicture}[thick,>=stealth,scale=0.3]
\draw(0,0) node[tre] (1) {}; \etq{1}
\draw(3,0) node[tre] (2) {}; \etq{2}
\draw(6,0) node[tre] (5) {}; 
\draw (5) node {\tiny $2$};
\draw(9,0) node[tre] (3) {}; \etq{3}
\draw(12,0) node[tre] (4) {}; \etq 4
\draw(1.5,2) node[tre] (a) {}; 
\draw(7.5,2) node[tre] (b) {}; 
\draw(4.5,4) node[tre] (c) {}; 
\draw(7.5,6) node[tre] (r) {}; 
\draw[->](r) --(c);
\draw[->](r)--(4);
\draw[->](c)--(a);
\draw[->](c)--(b);
\draw[->](a)--(1);
\draw[->](a)--(2);
\draw[->](b)--(5);
\draw[->](b)--(3);
\end{tikzpicture}
\end{center}
\caption{\label{fig:multree} A multilabelled tree.}
\end{figure}

This representation  of MUL-trees as phylogenetic networks allows us to translate to their setting the concepts developed so far. So, let $T$ be a weighted MUL-tree.  For every arc $e$ in it,  let $C(e)$ be the set of labels of its descendant leaves and $\kappa(e)=|C(e)|$ the number of different labels assigned to  descendant leaves of $e$. For every label $a\in \Sigma$,  we define its \emph{Fair Proportion} in $T$ as $FP_T(a)=\sum_{e: a\in C(e)} \omega(e)/\kappa(e)$: notice that now we split each $\omega(e)$ equally among the different labels of $e$'s descendant leaves, without taking into account their multiplicities, that is, how many leaves have any given label. Then, if, for every $X\subseteq \Sigma$, we define its \emph{MUL-Phylogenetic Diversity} in $T$ as
$$
mPSD_T(X)=\sum_{e:\, X\cap C(e)\neq \emptyset} \omega(e),
$$
Theorem \ref{th:main} applied to the phylogenetic network associated to $T$ implies that $FP_T$ is the Shapley Value of $mPSD_T$.
\end{remark}

\section{The unrooted subnet Shapley Value on a rooted phylogenetic network}

Consider the following two further coallitional games asssociated to a phylogenetic network $N$ on $\Sigma$: for every  $X\subseteq \Sigma$,
\begin{itemize}
\item The \emph{Cophenetic Value}  $CV_N(X)$ (cf.  \cite{copheneticd1,sokal.roth:62}) is 0 if $X=\emptyset$ and  the sum of the weights of the arcs that are ancestors of all leaves in $X$ otherwise:
$$
CV_N(\emptyset)=0\mbox{ and }CV_N(X)=\sum_{e:\, X\subseteq C(e)} \omega(e) \mbox{ if $X\neq\emptyset$}.
$$

\item The \emph{unrooted Phylogenetic Subnet Diversity}  $uPSD_N(X)$ is the difference
$$
uPSD_N(X)=rPSD_N(X)-CV_N(X)=\sum_{e:\, X\cap C(e)\neq \emptyset\atop X\not\subseteq C(e)} \omega(e).
$$
\end{itemize}
So, if $T$ is a phylogenetic tree, then $CV_T$ is equal to the usual cophenetic value of a set $X$ of leaves, that is, the total weight of the path going from the root of $T$ to $LCA_T(X)$, and  $uPSD_T$ is equal to Faith's unrooted Phylogenetic Diversity $uPD_T$ as recalled in \S 2.

For every $a\in \Sigma$, let its \emph{unrooted subnet Shapley Value} in $N$ be the Shapley Value of $uPSD_N$ on $a$,
$$
uSV_N(a)=\hspace*{-2ex} \sum_{a\in X\subseteq \Sigma} \frac{(|X|-1)!(|\Sigma|-|X|)!}{|\Sigma|!}\big(uPSD_N(X)-uPSD_N(X\setminus \{a\})\big),
$$
and its \emph{cophenetic Shapley Value} in $N$, the Shapley Value of $CV_N$ on $a$,
$$
cSV_N(a)=\hspace*{-2ex}\sum_{a\in X\subseteq \Sigma} \frac{(|X|-1)!(|\Sigma|-|X|)!}{|\Sigma|!}\big(CV_N(X)-CV_N(X\setminus \{a\})\big),
$$
By the additivity of Shapley Values, $rPSD_N=uPSD_N+CV_N$ implies that
$$
SV_N=uSV_N+cSV_N.
$$

Our goal is to obtain an expression for $uSV_N$ that generalizes to rooted phylogenetic networks the expression for $uPD_N$ on phylogenetic trees established in \cite{Haake1}. We do it using Theorem \ref{th:main} and the following expression for $cSV_N$.

\begin{lemma}\label{lem:coph}
For every weighted phylogenetic network $N$ on $\Sigma$ and for every $a\in \Sigma$,
$$
cSV_N(a) =\frac{1}{n}rPSD_N(\Sigma)-\sum_{e:\, a\notin C(e)}\frac{\omega(e)}{n-\kappa(e)}.
$$
\end{lemma}

\begin{proof}
Set $|\Sigma|=n$. To simplify the notations, we shall omit the subscripts $N$ in $CV_N$ and $cSV_N$. For every $\{a\}\subsetneq X\subseteq \Sigma$
$$
CV(X)-CV(X\setminus \{a\})  = 
\sum_{e:\, X\subseteq C(e)}\hspace*{-1ex} \omega(e)-\hspace*{-2ex}\sum_{e:\, (X\setminus \{a\})\subseteq C(e)}\hspace*{-2ex} \omega(e)=-\hspace*{-3ex}\sum_{e:\, (X\setminus \{a\})\subseteq C(e)\atop \quad a\notin C(e)}\hspace*{-1ex}\omega(e)
$$

while
$$
CV(\{a\})-CV(\emptyset)=\sum_{e:\, a\in C(e)}\hspace*{-1ex} \omega(e)=
rPSD_N(\Sigma)-\hspace*{-2ex}\sum_{e: a\notin C(e)}\hspace*{-1ex}\omega(e)
$$
Then,
$$
\begin{array}{rl}
cSV(a) & \displaystyle =\sum_{k=1}^n\frac{(k-1)!(n-k)!}{n!}\hspace*{-2ex}\sum_{|X|=k, a\in X} \big(CV(X)-CV(X\setminus \{a\})\big)\\
& \displaystyle =\frac{1}{n}(CV(\{a\})-CV(\emptyset))\\ &\qquad\displaystyle
+\sum_{k=2}^n\frac{(k-1)!(n-k)!}{n!}\hspace*{-2ex}\sum_{|X|=k, a\in X} \big(CV(X)-CV(X\setminus \{a\})\big)\\
\end{array}
$$
where, for every $k\geq 2$,
$$
\begin{array}{l}
\displaystyle \sum_{|X|=k, a\in X} (CV(X)-CV(X\setminus \{a\}))=
-\hspace*{-2ex}\sum_{|X|=k, a\in X}\sum_{e:\, (X\setminus \{a\})\subseteq C(e)\atop \quad a\notin C(e)}\hspace*{-1ex}\omega(e)
\\[2ex]
\displaystyle\qquad\quad
=-\hspace*{-2ex}\sum_{e:\, a\notin C(e)}\hspace*{-2ex}|\{Y\subseteq C(e)\mid |Y|=k-1\}|\cdot \omega(e)=-\hspace*{-2ex}\sum_{e:\, a\notin C(e)}\binom{\kappa(e)}{k-1}\omega(e)
\end{array}
$$
and therefore
$$
\begin{array}{l}
cSV(a)  \displaystyle =\frac{1}{n}\Bigg(rPSD_N(\Sigma)-\hspace*{-2ex}\sum_{e: a\notin C(e)}\hspace*{-1ex}\omega(e)\Bigg)\\
\displaystyle\qquad\qquad\qquad -\sum_{k=2}^n\Bigg(\frac{(k-1)!(n-k)!}{n!}\sum_{e:\, a\notin C(e)}\binom{\kappa(e)}{k-1}\omega(e)\Bigg)\\[2ex]
\qquad \displaystyle =\frac{1}{n}rPSD_N(\Sigma)-\sum_{k=1}^n\Bigg(\frac{(k-1)!(n-k)!}{n!}\sum_{e:\, a\notin C(e)}\binom{\kappa(e)}{k-1}\omega(e)\Bigg)\\[2ex]
\qquad  \displaystyle =\frac{1}{n}rPSD_N(\Sigma)-\sum_{e:\, a\notin C(e)}\Bigg(\omega(e)\sum_{k=1}^n\frac{(k-1)!(n-k)!}{n!}\binom{\kappa(e)}{k-1}\Bigg)\\[2ex]
 \qquad  \displaystyle =\frac{1}{n}rPSD_N(\Sigma)-\sum_{e:\, a\notin C(e)}\frac{\omega(e)}{n-\kappa(e)}
 \end{array}
$$
using again Lemma 6.15 in \cite{Steelbook}.
\qed
\end{proof}

Replacing the expressions for $cSV_N$ and $rCSV_N$ given in Theorem \ref{th:main} and the last lemma, respectively, in $uCSV_N=rCSV_N-cSV_N$, we obtain the following result:

\begin{theorem}\label{cor:unr}
For every weighted phylogenetic network $N$ on $\Sigma$ and for every $a\in \Sigma$,
$$
uSV_N(a) =\frac{1}{n}\sum_{e:\, a\in C(e)}\frac{n-\kappa(e)}{\kappa(e)}\cdot \omega(e)+\frac{1}{n}\sum_{e:\, a\notin C(e)}\frac{\kappa(e)}{n-\kappa(e)}\cdot \omega(e).
$$
\end{theorem}
It is not difficult to check that this expression agrees with the one given in \cite[Thm. 4]{Haake1} when $N$ is a rooted phylogenetic network.

\section{Conclusions}
In this note we have generalized to rooted phylogenetic networks two results on Shapley Values for phylogenetic trees: the equality of the rooted phylogenetic Shapley Value to the Fair Proportion and the simple expression of  the unrooted phylogenetic Shapley Value in terms of the weights and the number of descendant leaves of  arcs. 

We would like to call the reader's attention on the fact that Theorem \ref{th:main} is easily generalized to coallitional games $W:\mathcal{P}(\Sigma)\to \RR$ for which there exist a set $E$ and two mappings $C:E\to \mathcal{P}(\Sigma)$ and $\omega: E\to \RR$ such that 
$$
W(X)=\sum_{e:\, X\cap C(e)\neq \emptyset} \omega(e).
$$
For such a game $W$, the proof of Theorem \ref{th:main} \textsl{mutatis mutandis} shows that its Shapley value on $a\in \Sigma$  is simply
$$
SV_a(W)=\sum_{e:\, a\in C(e)} \frac{\omega(e)}{|C(e)|}.
$$

For instance, a Shapley Value of this type can be used to assess the importance of a question in an exam, one of the main goals of Item Response Theory \cite{IRT}, as follows. Let $\Sigma$ be the set of questions in an exam and let $E$ be the set of students taking this exam. We assume all questions in the exam to be  worth the same score. For every student $e$, let $C(e)$ be the set of questions correctly answered in her exam and set $\omega(e)=1/|E|$. For every set of questions $X$, let $W(X)=\sum_{e:\, X\cap C(e)\neq \emptyset} \omega(e)$, which is equal to the fraction of students that answered correctly some question in $X$.  Then, as we have just seen, the Shapley Value of this game on a given question $a$ is 
$$
SV_a(W)=\frac{1}{|E|}\sum_{e:\, a\in C(e)} \frac{1}{|C(e)|}.
$$
This Shapley Value measures the contribution of question $a$ to the global success in the exam; 
it increases with the number of students who answered the question correctly, but decreases with the grades they obtained.

 If different questions may have different scores, then it would be sensible to take as $\omega(e)$ the total score of the exam divided by $|E|$, in which case
$W(X)$, for a set of questions $X$, would be the average grade of the students who answered correctly some question in $X$. For another, recent use of the Shapley Value in the classification of items in an exam, see \cite{Luts}.
\medskip

\noindent\textbf{Acknowledgements.}
This research was partially supported by the Spanish Ministry of Economy and Competitiveness and the ERDF through project DPI2015-67082-P (MINECO/FEDER).  We thank G. Valiente and I. Garc\'\i a for their  helpful suggestions on several aspects of this paper.


\end{document}